\emergencystretch=\maxdimen
\documentclass[12pt,a4paper]{article} % A4 paper and 12pt font size
\usepackage{color}
\usepackage[utf8]{inputenc} % Polskie znaki. Potrzebne tylko w niektórych kompilatorach?
\usepackage[english]{babel}
\usepackage{mathtools,amsmath,amsfonts,amsthm,amssymb} % Matematyczne symbole, czcionki itp.
\usepackage{authblk} % Specyfikacja afiliacji.
\usepackage{indentfirst} % Wcinanie pierwszych paragrafów.
\usepackage[margin=3cm]{geometry} % Zmniejszenie marginesów do 3cm.
\usepackage{tocloft}
\usepackage{hyperref}
\hypersetup{
    colorlinks,
    citecolor=black,
    filecolor=black,
    linkcolor=black,
    urlcolor=black
}

\usepackage{cancel}

\newtheorem{thm}{Theorem}
\newtheorem{lem}[thm]{Lemma}
\newtheorem{prop}[thm]{Proposition}

\DeclareMathOperator{\Ran}{Ran}

\DeclareMathOperator{\spec}{spec}

\newcommand{\R}{\mathbb{R}}

\newcommand{\Z}{\mathbb{Z}}
\newcommand{\N}{\mathbb{N}}
\renewcommand{\i}{\mathrm{i}}
\newcommand{\e}{\mathrm{e}}

\begin{document}

\title{
Ground state energy of dense gases of strongly interacting fermions}

\author[1]{S\o ren Fournais}
\author[1]{B\l a\.zej Ruba}
\author[1]{Jan Philip Solovej}

\affil[1]{Department of Mathematical Sciences, University of Copenhagen, \protect\\
Universitetsparken 5, 2100 Copenhagen, Denmark }
%\affil[2]{Institution 2, \protect\\
%Address 2}

\date{\today}
\maketitle
\abstract{We study the ground state energy of a gas of $N$ fermions confined to a~unit box in $d$ dimensions. The particles interact through a 2-body potential with strength scaled in an $N$-dependent way as $N^{-\alpha}v$, where $\alpha\in \R$ and $v$ is a~function of positive type satisfying a mild regularity assumption. Our focus is on the strongly interacting case $\alpha<1-\frac2d$. We contrast our result with existing results in the weakly interacting case $\alpha>1-\frac2d$, and the transition happening 
at the mean-field scaling $\alpha=1-\frac2d$. Our proof is an adaptation of the bosonization technique used to treat the mean-field case. 
}
\section{Introduction}

One important goal in mathematical physics is to rigorously establish accurate approximations to ground energies of gases of quantum particles. For Fermi gases of low density $\rho$ this was initiated in \cite{LSS05}, where the leading contribution from the interactions to the energy was obtained. For gases in three-dimensional space it is of order $\rho^2$ and comes entirely from the potential between pairs of fermions with opposite spin. The same result with better error bounds (but stronger regularity assumptions on the potential) was derived in \cite{FGHP21,G23}. The currently best known upper bound \cite{G23} has errors of order $\rho^{\frac73}$, hence comparable with the conjectural next term \cite{HY57}. In the case of spinless (or spin-polarized) fermions, already the leading contribution of the interactions is expected to be of order $\rho^{\frac83}$. Using the method of cluster expansions this expected form of the interaction energy was established in~\cite{arxivLS23} as an upper bound. The method was also extended in \cite{arxivL23} to cover fermions with spin, leading to an upper bound with errors almost as small as in~\cite{G23}, but holding for a much larger class of potentials. There are also interesting results about gases at positive temperature \cite{S04,arxivLS23_pressure} and about one-dimensional gases~\cite{arxivARS22}.

There seems to be much less work on dense gases, presumably mainly due to the mathematical difficulty of the problem. Three early works on charged systems were \cite{GS94, B92, B93}. In particular, \cite{GS94} dealt with the thermodynamic limit of the Jellium model.  It has become popular to instead study systems of increasingly many particles confined to a fixed volume, and this we also do in the present work. This may be deemed somewhat unphysical, but it is not implausible that some conclusions of such investigations remain valid in more realistic models. Besides, simplified models provide a useful playground to develop new mathematical methods.

We consider the system of $N$ non-relativistic, interacting spinless fermions on the $d$-dimensional torus $\mathbb T^d = (\mathbb R / \mathbb Z)^d$ of dimension $d \geq 2$. We assume that interactions between particles are described by a two-body potential $v$ whose magnitude is modulated by the factor $N^{- \alpha}$, where $\alpha$ is a fixed numerical parameter. That is, we study the Hamiltonian
\begin{equation}
    H_N =  \sum_{i=1}^N (-\Delta_i) + N^{-\alpha} \sum_{1 \leq i < j \leq N} v(x_i-x_j)
    \label{eq:HN_def}
\end{equation}
on the Hilbert space $\mathcal H_N=L^2_a((\mathbb T^{d})^N)$ of square-integrable functions antisymmetric under permutations of the $N$ copies of $\mathbb T^d$. 

The lack of spin degrees of freedom in \eqref{eq:HN_def} does not make as much of a~difference as for dilute gases, and we make this assumption mostly to simplify the presentation. In Section \ref{sec:statement} we comment how the statement of our results should be modified to incorporate spin.

Since the ($N$-independent) length scales determined by the interaction are much larger than $N^{- \frac{1}{d}}$, which is the expected typical distance between the nearest particles, we are studying a very dense system. 

The orders of magnitude of the two terms in \eqref{eq:HN_def} are $N^{1+ \frac{2}{d}}$ and $N^{2- \alpha}$, respectively. This naive comparison suggests that the interaction potential is a small perturbation for $\alpha > 1 - \frac{2}{d}$, and that it plays a more dominant role for $\alpha < 1 - \frac{2}{d}$. The intermediate choice $\alpha = 1 - \frac{2}{d}$ is often called the mean field scaling. 

Here we consider the ground state energy $E_N$ (the infimum of the spectrum of~$H_N$). In order to frame the further discussion, let us mention two elementary bounds on $E_N$.

\begin{prop} \label{prop:trivial_bounds}
Suppose that the Fourier transform $\widehat v : (2 \pi \Z)^d \to \R$ of $v$ satisfies
\begin{equation}
    \widehat v(k) = \widehat v(-k) \geq 0 \text{ for } k \neq 0, \qquad \sum_{k \in (2 \pi \Z)^d} |k| \widehat v(k) < \infty.
\end{equation}
Introduce the quantity
\begin{equation}
    E_N^{(0)} = 
    \min_{\substack{p_1,\dots,p_N \in (2 \pi \mathbb Z)^d \\ \text{distinct}}} \sum_{i=1}^N |p_i|^2 + \frac{N^{2 - \alpha}}{2} \int_{\mathbb T^d} v - \frac{N^{1- \alpha}}{2} v(0).
    \label{eq:EN0}
\end{equation}
The minimum $E_N$ of the spectrum of the operator \eqref{eq:HN_def} satisfies
\begin{equation}
    E_N^{(0)} \leq E_N \leq E_N^{(0)} + c N^{1 - \alpha - \frac{1}{d}} \sum_{k \in (2 \pi \Z)^d} |k| \widehat v(k).
    \label{eq:trivial_bounds}
\end{equation}
for some positive constant $c$.
\end{prop}

We remark that if we replace $\mathbb T^d$ by $\mathbb T^d_L = (\mathbb R / L \mathbb Z)^d$, $N^{- \alpha}$ in \eqref{eq:HN_def} by $1$, and take $v$ to be the $L$-periodization of $V \in L^1(\mathbb R^d)$ with Fourier transform satisfying $\widehat V(k) \geq 0$ and $\int_{\mathbb R^d} |k| \widehat V(k) < \infty$, bounds analogous to \eqref{eq:trivial_bounds} give, in the limit $N \to \infty$, $L \to \infty$ with fixed $\rho = \frac{N}{L^d}$,
\begin{align}
    \liminf \frac{E_N}{L^d} & \geq c_1 \rho^{1 + \frac{2}{d}} + \frac{\rho^2}{2} \int_{\mathbb R^d} V(x) dx - \frac{\rho}{2} V(0), \\
    \limsup \frac{E_N}{L^d} & \leq c_1 \rho^{1 + \frac{2}{d}} + \frac{\rho^2}{2} \int_{\mathbb R^d} V(x)dx - \frac{\rho}{2} V(0) + c_2 \rho^{1 - \frac{1}{d}} \int_{\mathbb R^d} |k| \widehat V(k) dk ,
\end{align}
with explicit constants $c_1,c_2$ depending only on $d$. To the best of our knowledge, these are the best known bounds on the thermodynamic energy density for large particle density.

The upper bound in \eqref{eq:trivial_bounds} is obtained by calculating the expectation value $\langle \psi_0 | H_N \psi_0 \rangle$ in a ground state $\psi_0$ of the non-interacting gas. This may be seen as the first order of the perturbative expansion. A rigorous version of the perturbative expansion up to second order has been discussed in \cite{HPR20}, which focuses on $d=3$ and the mean field scaling $\alpha = \frac13$. The result of \cite{HPR20} applies only to small enough $v$. Accuracy of the expansion is better in the regime of weak interactions $\alpha > \frac13$. In~this case the proofs in \cite{HPR20} show that, at least for regular enough $v$, 
\begin{equation}
    E_N = \langle \psi_0 | H_N \psi_0 \rangle + O (N^{1- 2\alpha}).
    \label{eq:weak_interaction_EN}
\end{equation}
In fact, up to $o(N^{1- 2 \alpha})$ errors, $E_N- \langle \psi_0 | H_N \psi_0 \rangle$ is given by an explicit formula. Notice that $ N^{1 - 2 \alpha} \ll N^{1 - \alpha - \frac{1}{d}} = N^{\frac{2}{3}-\alpha}$. In this sense the upper bound in \eqref{eq:trivial_bounds} is sharp up to terms of lower order in $N$.

Further progress in understanding the ground state energy for $d=3$, $\alpha = \frac13$ was made in two series of works, \cite{BNPSS20,BNPSS21,BPSS23} and \cite{arXivCHN21,CHN23}. In \cite{BNPSS20} a new method of approximate bosonization (inspired by earlier works in physics literature, e.g.\ \cite{BP53,P53,GMB57,S57,SBFB57,L79,H94}) was developed to derive an upper bound on $E_N$. The bound was (at that time conjecturally) sharp up to $O(N^{\frac{1}{3}-\epsilon})$ errors, with some explicit $\epsilon >0$. The result was proved for potentials with positive and compactly supported Fourier transform. The first matching lower bound on $E_N$ was obtained in \cite{BNPSS21}, where it was also assumed that the Fourier transform of the potential has small $\ell^1$ norm. A~generalization to potentials as in Proposition \ref{prop:trivial_bounds} was obtained in~\cite{BPSS23}. Around the same time a slightly different approach to the bosonization method was proposed in \cite{arXivCHN21}, leading to similar results for the ground state energy. This new approach was further developed in \cite{CHN23}, where an upper bound for the ground state energy was derived for square-integrable potentials. This was a significant improvement because it allowed to treat the Coulomb potential. Some excited states of the gas were also considered in \cite{arXivCHN21, CHN22}.

The works \cite{BNPSS20,BNPSS21,BPSS23,arXivCHN21,CHN22} show in particular that in the mean field scaling (at least for $d=3$) neither the lower nor the upper bound in \eqref{eq:trivial_bounds} is sharp up to $o(N^{1- \alpha - \frac{1}{d}})$. By contrast, the upper bound is sharp for $\alpha > 1 - \frac{2}{d}$. The goal of this work is to complete this simple picture by showing that for strongly interacting systems, characterized by $\alpha < 1 - \frac{2}{d}$, the lower bound is sharp. We state this precisely in Theorem \ref{result}, which is our main result.

The proof of our result is based on the bosonization method. We would like to mention two differences between our treatment and \cite{BNPSS20,BNPSS21,BPSS23,arXivCHN21,CHN23}.

The authors of \cite{BNPSS20,BNPSS21,BPSS23} construct almost bosonic operators which create superpositions of many particle-hole pairs, with angular restrictions on momenta. The large number of pairs is needed to overcome the Pauli exclusion principle, while the angular restrictions (and linearization of the dispersion relation around the Fermi surface) guarantee that the created excitations have approximately definite kinetic energy. This allows to approximate the kinetic energy operator with a quadratic form in almost bosonic operators. By contrast, individual pair operators are used in \cite{arXivCHN21,CHN23}. This bypasses the need to linearize the dispersion relation. The individual pairs are not bosons, but since they are created with very small probability, they can be regarded as bosonic on the average. It turns out that this is enough to make the proofs work. In our case, $\alpha < 1- \frac{2}{d}$, the kinetic energy of pairs can be absorbed in error terms, at least for the purpose of proving Theorem \ref{result}. 

In the bosonization method, a Hamiltonian quartic in fermions is related to an operator quadratic in the almost bosonic operators. If these operators were exactly bosonic, one would be able to diagonalize the Hamiltonian using a Bogoliubov transformation. Bogoliubov transformations play a prominent role in the works \cite{BNPSS20,BNPSS21,BPSS23,arXivCHN21,CHN23}. They are not exact, but errors are managable -- essentially because the ground state of the effective quadratic Hamiltonian has (in a suitable average sense) $O(1)$ bosonic excitations. This is no longer true for strongly interacting systems. Therefore, instead of applying Bogoliubov transformations, we derive upper bounds by considering general trial states with $O(1)$ almost bosonic excitations. More precisely, we construct a linear map $\Phi$ from the Fock space of fictitious exact bosons to the Hilbert space of fermions and prove that it is almost isometric, up to errors which are $O(N^{-1 + \frac{1}{d}})$, but rapidly grow with the number of bosons. In a similar sense, $\Phi$ approximately intertwines between the effective bosonic Hamiltonian and $H_N$. We choose a state $f$ in the bosonic Fock space and take $\Phi f$ for the trial state. Then we calculate the energy in the limit $N \to \infty$ and only afterwards optimize over $f$. 

\section{Statement of the main result} \label{sec:statement}

We represent the two-body potential $v$ by Fourier series
\begin{equation}
    v(x) = \sum_{k \in (2 \pi \mathbb Z)^d} \widehat v(k) \e^{\i k x}.
    \label{eq:v_Fourier}
\end{equation}
$v$ has to be real and even, so $\widehat v(k) = \widehat v(-k)$ is real.

Potentials $v$ studied in this paper are bounded, so the differential operator $H_N$ in \eqref{eq:HN_def} is a self-adjoint operator on $\mathcal H_N = L^2_a((\mathbb T^d)^N)$. Its domain consists of all $f \in \mathcal H_N$ such that $\Delta f$ (understood in the sense of distributions) is also in $\mathcal H_N$. The spectrum of $H_N$ is discrete, bounded from below, and consists only of eigenvalues of finite multiplicity. The lowest eigenvalue of $H_N$ is called the ground state energy and denoted $E_N$.

In the absence of interactions (i.e.\ for $v=0$), $E_N$ is an eigenvalue of multiplicity~$1$ if and only if $N$ is of the form
\begin{equation}
    N = | \{ p \in (2 \pi \Z)^d \, | \, |p| \leq k_F \} |
    \label{eq:NF_def}
\end{equation}
for some $k_F > 0$, called the Fermi momentum.

\begin{thm} \label{result}
Let $d \geq 2$, $\alpha < 1 - \frac{2}{d}$. Suppose that the Fourier coefficients in~\eqref{eq:v_Fourier} satisfy
\begin{equation}
 \widehat v(0) \in \R, \qquad\widehat v(k) = \widehat v(-k) \geq 0 \text{ for } k \neq 0, \qquad \sum_{k \in (2 \pi \Z)^d} |k| \widehat v(k) < \infty.
    \label{eq:v_ass}
\end{equation}
Then for $N$ of the form \eqref{eq:NF_def} we have
\begin{equation}
    E_N = \sum_{\substack{p \in (2 \pi \Z)^d \\ |p| \leq k_F}} |p|^2 + \frac{N^{2-\alpha}}{2}
    \int_{\mathbb T^d} v - \frac{N^{1-\alpha}}{2} v(0) + o(N^{1-\alpha - \frac{1}{d}}).
\end{equation}
\end{thm}

Let us state a version of Theorem \ref{result} valid for particles with spin. Let $q$ be the number of spin states. For each spin value $\sigma \in \{ 1 , \dots, q \}$ we let the number of particles of spin $\sigma$ be
\begin{equation}
    N_\sigma = | \{ p \in (2 \pi \Z)^d \, | \, |p| \leq \lambda_\sigma k_F \} |,
    \label{eq:NF_def_sigma}
\end{equation}
where $\lambda_\sigma$ is fixed and $k_F \to \infty$ is the same for all $\sigma$. The total number of particles is $N = \sum_\sigma N_{\sigma}$. We consider the Hilbert space of square-integrable functions on $(\mathbb T^d)^N$ antisymmetric under permutations in $S_{N_1} \times \dots \times S_{N_q}$, on which we have the Hamiltonian
\begin{equation}
H_N = \sum_{i=1}^N (- \Delta_i) + N^{-\alpha} \sum_{1 \leq i < j \leq N} v_{\sigma(i) \sigma(j)} (x_i-x_j).
\end{equation}
Here $\sigma(i)$ is the spin of the $i$th particle, i.e.\ $\sigma(i)=1$ if $i \in \{ 1 , \dots, N_1 \}$, $\sigma(i)=2$ if $i \in \{ N_1 +1 , \dots, N_1+N_2 \}$ etc. The potential is taken to be a real matrix in the spin space satisfying $v(x)= v(-x)^{\mathrm{T}}$. In particular its Fourier transform $\widehat v(k)$ takes values in the hermitian matrices. Under the assumptions
\begin{equation}
    \widehat v(k) \text{ is a positive matrix for all } k \neq 0, \qquad \int |k| \| \widehat v(k) \| dk < \infty,
\end{equation}
a self-evident adjustment of the proof of Theorem \ref{result} shows that $E_N = \inf \spec H_N$ satisfies
\begin{equation}
    E_N = \sum_{\sigma=1}^q \sum_{\substack{p \in (2 \pi \mathbb Z)^d \\ |p| \leq \lambda_\sigma k_F}} |p|^2 + \frac{1}{2} \sum_{\sigma, \sigma'=1}^q N_\sigma N_{\sigma'} \int_{\mathbb T^d} v_{\sigma \sigma'} - \frac{1}{2} \sum_{\sigma=1}^q N_\sigma v_{\sigma \sigma}(0) + o (N^{1 - \alpha - \frac{1}{d}}).
\end{equation}

To avoid cluttered notation, we present the proof for the slightly less general statement in Theorem \ref{result}.

\section{Proof}

The $n$ particle Hilbert space is defined to be $\mathcal H_n =  L^2_a((\mathbb T^d)^n)$ -- square integrable functions on $(\mathbb T^d)^n$ anti-symmetric with respect to permutations of the $n$ coordinates. The direct sum $\mathcal H = \bigoplus_{n=0}^\infty \mathcal H_n$ is called the Fock space. 

We let $a_p^*,a_p$ be the creation and annihilation operators of a particle with momentum $p$. These are bounded operators on $\mathcal H$, adjoint to each other, such that
\begin{equation}
    a_p^* f (x_1, \dots, x_{n+1}) = \frac{1}{\sqrt{n+1}} \sum_{i=1}^{n+1} (-1)^{i-1} \e^{\i  p x_i} f(x_1, \dots, x_{i-1}, x_{i+1}, \dots, x_{n+1}) 
\end{equation}
for $f \in \mathcal H_n$. We have the canonical anticommutation relations 
\begin{equation}
    a_p a_q + a_q a_p = a_p^* a_q^* + a_q^* a_p^* = 0 , \qquad a_p a_q^* + a_q^* a_p = \delta_{p,q},
\end{equation}
where $\delta_{p,q}$ is the Kronecker delta. 

The plane wave state $\psi_0 \in \mathcal H_N$ is defined by
\begin{equation}
    \psi_0(x_1, \dots, x_N) = \frac{1}{\sqrt{N!}} \det(\e^{\i p_i x_j})_{i,j=1}^N,
\end{equation}
where $p_1, \dots, p_N$ are all the $p \in (2 \pi \Z)^d$ such that $|p| \leq k_F$. It satisfies $\| \psi_0 \|=1$ and
\begin{equation}
    a_p^* \psi_0 =0 \text{ for } |p| \leq k_F, \qquad a_p \psi_0 =0 \text{ for } |p| > k_F.
\end{equation}

The normally ordered kinetic energy operator is defined by\footnote{The series converges pointwise on the dense domain of all vectors obtained from $\psi_0$ by acting with a polynomial in $a_p^*$. By the sum of the series we mean the operator closure of the expression defined first on the said dense domain. The same remark applies to other similar expressions below.}
\begin{equation}
    :T: = \sum_{p \in (2 \pi \Z)^d} p^2 a_p^* a_p - \sum_{\substack{p \in (2 \pi \Z)^d \\ |p| \leq k_F}} p^2 = - \sum_{\substack{p \in (2 \pi \Z)^d \\ |p| \leq k_F}} p^2 a_p a_p^* + \sum_{\substack{p \in (2 \pi \Z)^d \\ |p| > k_F}} p^2 a_p^* a_p.
\label{eq:normal_ordered_T}
\end{equation}
By construction, $:T: \geq 0$ on $\mathcal H_N$.

For every $k \in (2 \pi \Z)^d$ let
\begin{equation}
    \rho_k = \sum_{p \in (2 \pi \Z)^d} a_{p-k}^* a_p.
\end{equation}
We note that $\rho_k^* = \rho_{-k}$ and that $\rho_k$ restrict to commuting operators on each $\mathcal H_n$ with norm $n$. We have $\rho_0=n$ on $\mathcal H_n$.

Let us define
\begin{equation}
    E_N^{(0)} = \sum_{\substack{p \in (2 \pi \Z)^d \\ |p| \leq k_F}} |p|^2 + \frac{N^{2-\alpha}}{2}
    \int_{\mathbb T^d} v - \frac{N^{1-\alpha}}{2} v(0).
\end{equation}
We have to prove that $E_N = E_N^{(0)} + o (N^{1 - \alpha - \frac{1}{d}})$.

With standard manipulations one can derive 
\begin{equation}
    H_N - E_N^{(0)} = :T: + \frac{N^{- \alpha}}{2} \sum_{\substack{k \in (2 \pi \Z)^d \\ k \neq 0}} \widehat v(k) \rho_k^* \rho_k,
\end{equation}
in which restriction of the right hand side to $\mathcal H_N$ is implicitly understood. From this we see immediately that, under assumptions of Theorem \ref{result},
\begin{equation}
    E_N - E_N^{(0)} \geq 0.
\end{equation}
Therefore, we need only an upper bound on $E_N - E_N^{(0)}$. One upper bound is obtained by considering the trial vector $\psi_0$.
\begin{equation}
    E_N \leq \langle \psi_0 | H_N \psi_0 \rangle = E_N^{(0)} + \frac{N^{- \alpha}}{2} \sum_{\substack{k \in (2 \pi \Z)^d \\ k \neq 0}} |C_k| \widehat v(k),
\label{eq:trial_state_pw}
\end{equation}
where $|C_k|$ is the number of elements of the crescent set $C_k$:
\begin{equation}
    C_k =  \{ p \in (2 \pi \Z)^d \, | \, |p| \leq k_F, \, |p+k| > k_F \}.
\end{equation}
By Lemma \ref{lem:Ck_size} below, \eqref{eq:trial_state_pw} gives $E= E_N^{(0)} + O(N^{1 - \alpha - \frac{1}{d}})$. A better trial state is needed to replace $O(N^{1 - \alpha - \frac{1}{d}})$ with $o(N^{1 - \alpha - \frac{1}{d}})$.

\begin{lem} \label{lem:Ck_size}
There exist positive constants $c_1,c_2, c_3$, depending only on $d$, such that for all $k \in (2 \pi \Z)^d$ and for all $k_F \geq c_3$ we have
    \begin{equation}
        c_1 k_F^{d-1} \min \{ |k|, k_F \} \leq |C_k| \leq c_2 k_F^{d-1} \min \{ |k|, k_F \}.
\label{eq:crescent_bounds}
    \end{equation}
\end{lem}
\begin{proof}
If $d=1$ (which is not needed in this paper), one can check \eqref{eq:crescent_bounds} by a~direct computation. The estimate is also clear for $k =0$, because then $C_k = \emptyset$. We assume that $d \geq 2$ and $ k \neq 0$ (so $|k| \geq 2 \pi$). 

Let $B(K) = \{p \in (2 \pi \Z)^d \, | \, |p| \leq K  \}$. One has
 \begin{equation}
     |B(K)|= (2 \pi)^{-d} |\mathbb S^{d-1}| |K|^d + o(K^{d-1}),
 \end{equation}
where $\mathbb S^{d-1}$ is the $(d-1)$-dimensional unit sphere and $|\mathbb S^{d-1}|$ is its volume.

First we prove the upper bound on $|C_k|$. Consider the case $|k| \leq k_F$. Then
    \begin{equation}
        C_k \subset B(k_F) \setminus B(k_F-|k|), \qquad B(k_F-|k|) \subset B(k_F),
    \end{equation}
    and therefore we have
    \begin{align}
|C_k| &\leq (2 \pi)^{-d}|\mathbb S^{d-1}| (k_F^d - (k_F-|k|)^d) + o(k_F^{d-1}) \\
&\leq (2 \pi)^{-d}|\mathbb S^{d-1}| d \, k_F^{d-1} |k| + o (k_F^{d-1}). \nonumber
    \end{align}
For large enough $k_F$ (independent of $k$) we can drop the $o (k_F^{d-1})$ error at the price of slightly increasing the constant in the first term.

In the case $|k| > k_F$ we use
\begin{equation}
    |C_k| \leq N \leq c k_F^d.
\end{equation}

Next we prove the lower bound for $|k| \leq k_F$. Let $G$ be the group of linear isometries of $(2 \pi \Z)^d$, $G = \mathrm{GL}(d, \Z) \cap \mathrm{O}(d)$. Considering how $G$ acts on the set $\{ (\pm 1, 0, \dots), (0, \pm 1, 0, \dots), \dots \}$ we see that $|G|= 2^{\frac{d(d+1)}{2}}$. Every $G$-orbit in $\mathbb S^{d-1}$ meets the set
\begin{equation}
D=\{ (\omega_1,\dots,\omega_d) \in \mathbb S^{d-1} \, | \, 0 \leq \omega_1 \leq \omega_2 \leq \dots \leq \omega_d \},
\end{equation}
and for any two $\omega,\omega' \in D$ we have $\omega \cdot \omega' \geq \frac{1}{\sqrt{d}}$. 

Let $p \in B(k_F) \setminus B\Big(k_F - \frac{1}{\sqrt{d}}|k| \Big)$. Choose $g \in G$ such that $p \cdot gk \geq \frac{1}{\sqrt{d}} |p| |k|$. Then
\begin{align}
   |p + gk|^2 & \geq |p|^2 + \frac{2}{\sqrt{d}} |p| |k| + |k|^2 > \Big(k_F- \frac{1}{\sqrt{d}} |k| \Big)^2 + \frac{2|k|}{\sqrt{d}} \Big(k_F - \frac{1}{\sqrt{d}} |k| \Big) + |k|^2 \nonumber \\
   & = k_F^2 + \frac{d-1}{d} |k|^2 > k_F^2,
\end{align}
so $p \in C_{g k}$. We proved that
\begin{equation}
    B(k_F) \setminus B \Big( k_F - \frac{1}{\sqrt{d}} |k| \Big) \subset \bigcup_{g \in G} C_{g k}.
\end{equation}
Since $|C_{gk}|= |C_k|$, we have
\begin{equation}
    |C_k| \geq 2^{- \frac{d(d+1)}{2}} \left| \bigcup_{g \in G} C_{g k} \right| \geq 2^{- \frac{d(d+1)}{2}} \left| B(k_F) \setminus B \Big( k_F - \frac{1}{\sqrt{d}} |k| \Big) \right|.
\end{equation}
Next, we have
\begin{align}
    \left| B(k_F) \setminus B \Big( k_F - \frac{1}{\sqrt{d}} |k| \Big) \right| & = (2 \pi)^{-d} |\mathbb S^{d-1}| \Big( k_F^d - \Big( k_F - \frac{1}{\sqrt{d}} |k| \Big)^d \Big) + o(k_F^{d-1}) \\
     & \geq (2 \pi)^{-d} |\mathbb S^{d-1}| d \, \Big( k_F- \frac{|k|}{\sqrt{d}} \Big)^{d-1} \frac{|k|}{\sqrt{d}} +o(k_F^{d-1})  . \nonumber
\end{align}
Again restricting to large enough $k_F$, this may be lower bounded by $c k_F^{d-1} |k|$. 

If $|k| > k_F$, we have $C_k \cup C_{-k} = B(k_F)$, and hence
\begin{equation}
    |C_k| \geq \frac12 N \geq c k_F^{d}.
\end{equation}
\end{proof}

In what follows, we take $k_F$ large enough that \eqref{eq:crescent_bounds} holds.

It is useful to split $\rho_k$ with $k \neq 0$ as follows:
\begin{equation}
    \rho_k = b_k + b_{-k}^* + d_k,
\end{equation}
where we introduce (all $p$ here and in other sums below are in $(2 \pi \Z)^d$):
\begin{align}
    b_{-k}^* &= \sum_{\substack{|p| \leq k_F \\ |p-k| > k_F}} a_{p-k}^* a_p, \label{eq:bstar} \\
    b_k &= \sum_{\substack{|p| > k_F \\ |p-k| \leq k_F}} a_{p-k}^* a_{p}, \\
    d_k & = \sum_{|p|, |p-k| \leq k_F } a_{p-k}^* a_p + \sum_{|p|, |p-k| > k_F } a_{p-k}^* a_p .
\end{align}
Operator $b_{-k}^*$ is indexed by $-k$ so that $b_k$ and $b_k^*$ are adjoint to each other.

We have the commutation rules
\begin{gather}
    [b_k,b_q]=[b_k^*,b_q^*]= 0 , \nonumber \\
    [b_k , b_q^*] = |C_k| \delta_{k,q} + :[b_k,b_q^*]:, \label{eq:commutator_decomp}
\end{gather}
where we introduced
\begin{equation}
    :[b_k,b_q^*]: = - \sum_{\substack{|p|, |p+q-k| \leq k_F \\ |p+q| > k_F}} a_p a_{p+q-k}^* - \sum_{\substack{|p| \leq k_F \\ |p+q|, |p+k| > k_F}} a_{p+q}^* a_{p+k}.
\end{equation}
We will not use it, but we note that $- :[b_k,b_k^*]: \geq 0$.

The basis of the bosonic approximation is that $d_k$ and $:[b_k,b_q^*]:$ annihilate $\psi_0$, and they are small when acting on states close to $\psi_0$. The closeness is measured by the number of excitations operator
\begin{equation}
    \mathcal N = \frac12 \sum_{|p| \leq k_F} a_p a_p^* + \frac12 \sum_{|p| > k_F} a_p^* a_p ,
    \label{eq:exc_number}
\end{equation}
see Lemma \ref{lem:basic_est} below. We remark that the two terms in \eqref{eq:exc_number} are equal on $\mathcal H_N$.

\begin{lem} \label{lem:basic_est}
Let $\psi \in \mathcal H$.
\begin{enumerate}
    \item If $\| \mathcal N^{\frac12} \psi \| < \infty$, then $\| b_k \psi \| \leq |C_k|^{\frac{1}{2}}  \| \mathcal N^{\frac12} \psi \|$, $\| b_k^* \psi \| \leq |C_k|^{\frac{1}{2}}  \| (\mathcal N + 1)^{\frac12} \psi \|$. 
\item If $\| \mathcal N \psi \| < \infty$, then $\| :[b_k,b_q^*]: \psi \| \leq 2 \| \mathcal N \psi \|$ and $\| d_k \psi \| \leq 2 \| \mathcal N \psi \|$. 
\end{enumerate}
\end{lem}
\begin{proof}
Standard applications of Cacuhy-Schwarz inequality.
\end{proof}

It is convenient to introduce normalized approximately bosonic operators
\begin{equation}
    \varphi_k = |C_k|^{-\frac12} b_k, \quad     \varphi_k^* = |C_k|^{-\frac12} b_k^*, \quad :[\varphi_k,\varphi_q^*]: = |C_k|^{-\frac12} |C_q|^{-\frac12} :[b_k,b_q^*]:. 
    \label{eq:normalized_bosons}
\end{equation}
By Lemma \ref{lem:Ck_size}, this is a meaningful definition as $|C_k| \neq 0$ for all $k$ if $k_F$ is large enough. 

Informally, what follows is based on the approximation
\begin{equation}
    \rho_k \approx |C_k|^{\frac12} (\varphi_k + \varphi_{-k}^*), \qquad [\varphi_k,\varphi_q^*] \approx \delta_{k,q}.
\end{equation}

Let $\mathcal F$ be the bosonic Fock space over $\ell^2((2 \pi \Z)^d \setminus \{ 0 \})$. We denote the vacuum vector of $\mathcal F$ by $\Omega$ and the creation and annihilation operators corresponding to elements of $(2 \pi \Z)^d \setminus \{ 0 \}$ by $e_k^*, e_k$. Most of the time we will work in the dense subspace $\mathcal D \subset \mathcal F$ of vectors obtained from $\Omega$ by acting with a polynomial in $e_k^*$. 

If $S \subset (2 \pi \Z)^d \setminus \{ 0 \}$ is finite, we let $\mathcal D^S \subset \mathcal D $ be the subspace obtained by acting on $\Omega$ with polynomials in $e_k^*$ with the restriction $k \in S$. By definition of $\mathcal D$, every element of $\mathcal D$ is in some $\mathcal D^S$:
\begin{equation}
    \mathcal D = \bigcup_{S} \mathcal D^S.
\end{equation}
 
%Let $\mathcal F$ be the prehilbert space of polynomials in variables $(e_k^*)_{k \neq 0}$, equipped with the scalar product 
%\begin{equation}
%\Big \langle \prod_k e_k^{* \mu_k} | \prod_k e_k^{* \nu_k} \Big \rangle = \prod_k \delta_{\mu_k, \nu_k} \, \mu_k!.
%\end{equation}
%We use $*$ in the names of variables $e_k^*$ because they play the role of bosonic creators. The corresponding annihilators will be denoted by $e_k$. 

For a natural number $m$ we let $\mathcal D_{m} \subset \mathcal D$ be the $m$ particle subspace, i.e. the linear span of vectors $e_{k_1}^* \cdots e_{k_m}^* \Omega$. We introduce also
\begin{equation}
\mathcal D_{\leq m} = \bigoplus_{i=0}^m \mathcal D_i, \qquad \mathcal D^S_{ m} = \mathcal D^S \cap \mathcal D_{ m}, \qquad \mathcal D^S_{\leq m} = \mathcal D^S \cap \mathcal D_{\leq m}.
\end{equation}
We note that
\begin{equation}
 \dim(\mathcal D_{\leq m}^S) = \binom{m + |S|}{m}, \quad    \dim(\mathcal D_{ m}^S) = \binom{m + |S|-1}{m}, \quad \mathcal D= \bigcup_{m , S} \mathcal D^S_{\leq m}.
\end{equation}

We define a linear map $\Phi : \mathcal D \to L^2_{\mathrm{a}}(\mathbb T^{dN})$ by
\begin{equation}
    \Phi(e_{k_1}^* \cdots e_{k_m}^* \Omega) = \varphi_{k_1}^* \cdots \varphi_{k_m}^* \psi_0.
\end{equation}
Equivalently, $\Phi$ is uniquely determined by the conditions 
\begin{equation}
 \Phi(\Omega)=\psi_0, \qquad  \Phi e_k^* = \varphi_k^* \Phi.   
\end{equation}
We note that $\Phi$ depends on $k_F$.

We let $\Phi^S_m$ and $\Phi^S_{\leq m}$ be the restrictions of $\Phi$ to $\mathcal D^S_m$ and $\mathcal D^S_{\leq m}$, respectively. These are linear operators on finite-dimensional Hilbert spaces, so they are bounded. 

Let $\mathcal D^*$ be the space of antilinear functionals on $\mathcal D$ (not necessarily continuous). We have a canonical embedding $\mathcal D \to \mathcal D^*$ given by $f \mapsto \langle \cdot | f \rangle$. If $f \in \mathcal D^*$ and $g \in \mathcal D$, we denote the evaluation of $f$ at $g$ by $\langle g|f \rangle$ and put $\langle f|g \rangle := \overline{\langle g|f \rangle}$. 

We define the adjoint of $\Phi$ to be the map
\begin{equation}
    \Phi^* : L^2_{\mathrm{a}} (\mathbb T^{dN}) \to \mathcal D^*, \qquad \langle g | \Phi^* f \rangle = \langle \Phi g | f \rangle \text{ for } g \in \mathcal D, \ f \in L^2_{\mathrm{a}}(\mathbb T^{dN}).
\end{equation}
Our $\Phi^*$ is an extension of the standard operator adjoint. It is defined on all of $L^2_{\mathrm{a}} (\mathbb T^{dN})$, at the price of being valued in the (very large) space $\mathcal D^*$. We expect that working with $\mathcal D^*$ is not strictly necessary, but it is convenient because it allows to not worry about domain issues in the (essentially algebraic) calculations below.

Note that if $g \in \mathcal D^S_{\leq m}$, then
\begin{equation}
    \langle g | \Phi^* f \rangle = \langle g | (\Phi^S_{\leq m})^* f \rangle,
\end{equation}
and $(\Phi^S_{\leq m})^*$ is a bounded operator $L^2_{\mathrm{a}}(\mathbb T^{dN}) \to \mathcal D^{S}_{\leq m}$.

\begin{lem} \label{lem:isometric_Phi}
We have a bound
\begin{equation}
    |\langle e^*_{k_m} \cdots e_{k_1}^*\Omega |  (\Phi^* \Phi -1)  e^*_{q_m} \cdots e^*_{q_1} \Omega \rangle| \leq c \binom{m}{2} m! k_F^{1-d}  ,\label{eq:Phi_iso}
\end{equation}
where $c$ is a universal constant. In particular we have 
\begin{equation}
    \| (\Phi^{S}_{\leq m})^* \Phi_{\leq m}^S-1 \| \leq c \binom{m+|S|-1}{m} \binom{m}{2} m! k_F^{1-d}.
\end{equation}
Hence (assuming $d \geq 2$) for any $f \in \mathcal D$ we have
\begin{equation}
    \lim_{k_F \to \infty} \| \Phi f \| = \| f \|.
\end{equation}
\end{lem}
\begin{proof}
For every $m \geq 0$ we define
\begin{align}
\epsilon_m(k_m, \dots, k_1 ; q_m , \dots , q_1) = & \langle e^*_{k_m} \cdots e^*_{k_1} \Omega |  (\Phi^* \Phi-1) e^*_{q_m} \cdots e_{q_1}^* \Omega \rangle   .
\end{align}
The goal is to show
\begin{equation}
    |\epsilon_m(k_m, \dots, k_1 ; q_m , \dots , q_1)| \leq c \binom{m}{2} m! k_F^{1-d}. 
    \label{eq:eps_bound}
\end{equation}
We proceed by induction on $m$. Cases $m=0$ and $m=1$ are trivial. Using the Leibniz rule for commutators and decomposition \eqref{eq:commutator_decomp} we derive
\begin{align}
\label{eq:Phi_ind}    \langle &  e^*_{k_{m+1}}  \cdots e^*_{k_1} \Omega |\Phi^* \Phi e^*_{q_{m+1}} \cdots e^*_{q_1} \Omega \rangle = \langle \varphi_{k_{m+1}} ^* \cdots \varphi_{k_1}^* \psi_0 | \varphi_{q_{m+1}}^* \cdots \varphi_{q_1}^* \psi_0 \rangle \\
    = & \sum_{i=1}^{m+1} \langle \varphi_{k_m}^* \cdots \varphi_{k_1}^* \psi_0 | \varphi_{q_m}^* \cdots  [\varphi_{k_{m+1}}, \varphi_{q_i}^*] \cdots \varphi_{q_1}^* \psi_0  \rangle \nonumber \\
 =&  \sum_{i=1}^{m+1} \Bigg[ \delta_{q_i}^{k_{m+1}} \langle e_{k_m}^* \cdots e^*_{k_1} \Omega | \Phi^* \Phi e_{q_{m+1}}^* \cdots e_{q_{i+1}}^* e_{q_{i-1}}^* \cdots e_{q_1}^* \Omega \rangle \nonumber \\
 & + \langle \varphi_{k_m}^* \cdots \varphi_{k_1}^* \psi_0 | \varphi_{q_{m+1}}^* \cdots :[\varphi_{k_{m+1}}, \varphi_{q_i}^*]: \cdots \varphi_{q_1}^* \psi_0 \rangle \Bigg]. \nonumber
\end{align}
In the first term of the summand we use $\Phi^* \Phi = 1 + (\Phi^* \Phi-1)$ and the identity
\begin{align}
& \sum_{i=1}^{m+1}  \delta_{q_i}^{k_{m+1}}\langle e_{k_m}^* \cdots e^*_{k_1} \Omega |  e_{q_{m+1}}^* \cdots e_{q_{i+1}}^* e_{q_{i-1}}^* \cdots e_{q_1}^* \Omega \rangle \\
&= \langle e_{k_{m+1}}^* \cdots e^*_{k_1} \Omega |  e_{q_{m+1}}^* \cdots  e_{q_1}^* \Omega \rangle . \nonumber
\end{align}
This allows to rewrite \eqref{eq:Phi_ind} in the form
\begin{align}
   & \epsilon_{m+1}(k_{m+1}, \dots, k_1 ; q_{m+1}, \dots, q_1) - \sum_{i=1}^{m+1} \delta^{k_{m+1}}_{q_i} \epsilon_m(k_m , \dots, k_1 ; q_{m+1} , \dots , \cancel{q_{i}}, \dots, q_1) \nonumber \\
    = & \sum_{i=1}^{m+1} \langle \varphi_{k_m}^* \cdots \varphi_{k_1}^* \psi_0 | \varphi_{q_{m+1}}^* \cdots :[\varphi_{k_{m+1}}, \varphi_{q_i}^*]: \cdots \varphi_{q_1}^* \psi_0 \rangle. 
\end{align}
Using Lemmas \ref{lem:Ck_size} and \ref{lem:basic_est} we bound the right hand side by $c \, (m+1)! m k_F^{1-d}$. Using the induction hypothesis, we get \eqref{eq:eps_bound} for $m+1$.

For an $n \times n$ matrix $T$ we have $\| T \| \leq n \max_{i,j} |T_{ij}|$. Since monomials in $e^*_k$ have norms at least $1$ and form an orthogonal basis, this trivial estimate gives
\begin{equation}
    \| (\Phi^{S}_{m})^* \Phi_{m}^S-1 \| \leq c \binom{m+|S|-1}{m} \binom{m}{2} m! k_F^{1-d}.
    \label{eq:Phi_isom_estimate_fixed_deg}
\end{equation}
The orthogonal decomposition $\mathcal D^S_{\leq m} = \bigoplus_{i=0}^m \mathcal D^S_{i}$ is preserved by $(\Phi^{S}_{\leq m})^* \Phi_{\leq m}^S$, so~\eqref{eq:Phi_isom_estimate_fixed_deg} is satisfied also for $(\Phi^{S}_{\leq m})^* \Phi_{\leq m}^S-1$ (because the right hand side of \eqref{eq:Phi_isom_estimate_fixed_deg} is an increasing function of $m$).
\end{proof}

Let us split $H_N= E_N^{(0)} +H^{(1)}+H^{(2)}$, where
\begin{align}
H^{(1)}  =& \frac{N^{-\alpha}}{2} \sum_{k \neq 0} |C_k| \widehat v(k) (\varphi_k^* + \varphi_{-k})(\varphi_{-k}^* + \varphi_k) , \\
H^{(2)}  =&  :T: + \frac{N^{-\alpha}}{2} \sum_{k \neq 0} \widehat v(k) \left( (b_k^* + b_{-k}) d_k + d_{-k} (b_{-k}^* +b_k) + d_{-k}d_k \right). \nonumber
\end{align}
Furthermore, let us define
\begin{align}
H^B &= \frac{N^{-\alpha}}{2} \sum_{k \neq 0} |C_k| \widehat v(k)  (e_k^*+e_{-k})(e_{-k}^* + e_k), \\    
\widetilde H^B & = \sum_{k \neq 0} |k| \widehat v(k)  (e_k^*+e_{-k})(e_{-k}^* + e_k).\nonumber
\end{align}
The operators $H^B, \widetilde H^B$ are self-adjoint on $\mathcal F$ with spectrum $[0, \infty)$ (unless $v=0$ in which case $H^B = \widetilde H^B=0$). The domain $\mathcal D$ is a core of $H^B, \widetilde H^B$.

We will treat $H^{(2)}$ as an error term and $H^B$ as an approximation to $ H^{(1)}$. The operator $\widetilde H^B$ does not depend on $k_F$ and by Lemma \ref{lem:Ck_size} we have
\begin{equation}
    0 \leq H^B \leq c N^{1- \alpha - \frac{1}{d}} \widetilde H^B, \qquad \inf_{\substack{f \in \mathcal D \\ \| f \| =1}} \langle f | \widetilde H^B f \rangle = 0,
    \label{eq:HB_upper}
\end{equation}
for some $c >0$.

\begin{lem} \label{lem:H1_intertwine}
    We have a bound
    \begin{equation}
        \| (\varphi_k \Phi -\Phi e_k)e^*_{q_1} \cdots e^*_{q_m} \Omega \| \leq c k_F^{1-d} \sqrt{(m-1)!} \binom{m}{2},
\label{eq:ani_intertwine_monom}
    \end{equation}
    and therefore
    \begin{equation}
        \| \varphi_k \Phi_{\leq m}^S - \Phi^S_{\leq m} e_k \| \leq c k_F^{1-d} \binom{m+|S|-1}{m}^{\frac12} \sqrt{(m-1)!} \binom{m}{2}. 
        \label{eq:ani_intertwine_cut}
    \end{equation}
    Similarly,
    \begin{equation}
   \| (\varphi_k^* + \varphi_{-k})(\varphi_k + \varphi_{-k}^*) \Phi^{S}_{\leq m} - \Phi^{S \cup \{ k, - k \}}_{\leq m+2} (e_k^* +e_{-k})(e_k + e_{-k}^*) \| \leq c_{m}^{|S|} k_F^{1-d},
    \label{eq:quad_intertwine}
    \end{equation}
   with some explicit constant $c_{m}^{|S|}$ depending only on $m$ and $|S|$. Hence
    \begin{equation}
        \| H^{(1)} \Phi_{\leq m}^S - \Phi H^B \| \leq c_{m}^{|S|} N^{- \alpha} \sum_k |k| |\widehat v(k)|.
        \label{eq:H1_intertwine}
    \end{equation}
\end{lem}
\begin{proof}
We have the identity
\begin{equation}
    (\varphi_k \Phi -\Phi e_k)e^*_{q_1} \cdots e^*_{q_m} \Omega = \sum_{i=1}^m \varphi_{q_1}^* \cdots :[\varphi_k, \varphi_{q_i}^*]: \cdots \varphi_{q_m}^* \psi_0.
    \end{equation}
Now bound term by term using Lemmas \ref{lem:Ck_size}, \ref{lem:basic_est}. This proves \eqref{eq:ani_intertwine_monom}. Then \eqref{eq:ani_intertwine_cut} follows because for an $n \times n$ matrix with columns of norm $c$ has norm at most $\sqrt{n}c$. Bound \eqref{eq:quad_intertwine} is derived from \eqref{eq:ani_intertwine_cut} and the identity $\varphi_k^* \Phi_{\leq m}^S = \Phi^{S \cup \{ k \}}_{\leq m+1} e_k^*$ using the triangle inequality. Then \eqref{eq:quad_intertwine} implies \eqref{eq:H1_intertwine}.
\end{proof}

\begin{lem} \label{lem:H2_est}
Fix $K \in [2 \pi,k_F]$ and let $S = \{ k \in (2 \pi \Z)^d \setminus \{ 0 \} \, | \, |k| \leq K \}$. For every unit vector $\psi \in \Ran(\Phi_{\leq m}^S)$ we have
\begin{align}
| \langle \psi| H^{(2)} \psi \rangle| \leq & (2k_F K+K^2)m + \frac{N^{-\alpha}}{2} \sum_{k \neq 0} |\widehat v(k)| (8 m (m+1)^{\frac12} |C_k|^{\frac12} +4m^2).
\nonumber
\end{align}
If $\alpha < 1 - \frac{2}{d}$, the right hand side is $o (N^{1 - \alpha - \frac{1}{d}})$ for $k_F \to \infty$.
\end{lem}
\begin{proof}
If $p \in C_k$ with $|k| \leq K$, then $|p+k|^2 - |p|^2 \leq 2k_FK +K^2$. Since $\psi$ has at most $m$ particles and holes created by $b_k^*$ with such $k$, the kinetic term is at most $(2k_F K+K^2) m$. For the other terms we use Lemmas \ref{lem:Ck_size} and \ref{lem:basic_est}.
\end{proof}

We are now ready to finalize.

\begin{proof}[Proof of Theorem \ref{result}]
Let $f \in \mathcal D$. Then for some $m \in \N$ and $K, S$ as in \ref{lem:H2_est} we have $f \in \mathcal D^S_{\leq m}$. By Lemmas \ref{lem:isometric_Phi}, \ref{lem:H1_intertwine}, \ref{lem:H2_est} and by \eqref{eq:HB_upper} we have
\begin{align}
  E_N \leq   \frac{\langle \Phi f | H_N \Phi f \rangle}{\langle \Phi f | \Phi f \rangle} & = E_N^{(0)} + \langle f | H^B f \rangle + o (N^{1- \alpha - \frac{1}{d}}) \\
   \nonumber & \leq E_N^{(0)} + c N^{1 - \alpha - \frac{1}{d}} \langle f | \widetilde H^B f \rangle + o (N^{1- \alpha - \frac{1}{d}}) . 
\end{align}
This proves that
\begin{equation}
    \limsup_{k_F \to \infty} \frac{E_N - E_N^{(0)}}{N^{1 - \alpha - \frac{1}{d}}} \leq c \langle f | \widetilde H^B f \rangle.
\end{equation}
Now optimize over $f$.
\end{proof}

\section*{Acknowledgments}

We would like to thank Johannes Agerskov, Martin Ravn Christiansen, Jan Derezinski and Robert Seiringer for helpful discussions. 
This work was
partially supported by the Villum Centre of Excellence for the
Mathematics of Quantum Theory (QMATH) with Grant No.10059.
SF was partially funded by the European Union. Views and opinions expressed are however those of the authors only
and do not necessarily reflect those of the European Union or the European Research Council. Neither
the European Union nor the granting authority can be held responsible for them.

\end{document}